\documentclass[11pt]{amsart}
\usepackage{
  algpseudocode,
  algorithm,
  amsmath,
  amsfonts,
  amssymb,
  epsfig,
  epstopdf,
  geometry,
  graphicx,
  xfrac,
  latexsym,
  setspace,
  url,
  hyperref,
  epstopdf
  }
\newtheorem{theorem}{Theorem}

\newtheorem{example}[theorem]{Example}

\title{On a Stackelberg Subset Sum Game}
\author{
   Ulrich Pferschy,
   Gaia Nicosia, 
   Andrea Pacifici
   }
\date{}                                           

\def\Re{\mathbb{R}}

\def\np{$\mathcal{NP}$}
\def\ll{$\mathcal L$}
\def\ff{$\mathcal F$}
\def\tw{\tilde w}
\def\eps{\varepsilon}

\thanks{U. Pferschy: University of Graz, Austria. Email:
   \href{mailto:pferschy@uni-graz.at} {\texttt{
   pferschy@uni-graz.at}}}
\thanks{G. Nicosia: Universit\`a Roma Tre, Italy. Email:
   \href{mailto:nicosia@ing.uniroma3.it} {\texttt{
     nicosia@ing.uniroma3.it}}}
\thanks{A. Pacifici: Universit\`a di Roma Tor Vergata, Italy. Email:
   \href{mailto:andrea.pacifici@uniroma2.it} {\texttt{
   andrea.pacifici@uniroma2.it}}}

\begin{document}
\maketitle

\begin{abstract}
This contribution deals with a two-level 
discrete decision problem, 
a so-called Stackelberg strategic  game: A Subset Sum setting is addressed 
with a set $N$ of items with given integer weights. 
One distinguished player, the leader \ll, may alter the weights of the items 
in a given subset $L\subset N$, and a second player, the follower \ff, selects 
a solution $A\subseteq N$ in order to utilize a bounded resource in the best 
possible way. Finally, the leader receives a payoff from those items of its 
subset $L$ that were included in the overall solution $A$, chosen by the 
follower. 
We assume that \ff{} applies a publicly known, simple, heuristic algorithm 
to determine its solution set, which avoids having to solve \np-hard problems.

Two variants of the problem are considered, depending on whether \ll{} is 
able to control (i.e., change) the weights of its items $(i)$ in the objective 
function or $(ii)$ in the bounded resource constraint. 
The leader's objective is the maximization of the overall weight reduction, 
for the first variant, or the maximization of the weight increase for the 
latter one. In both variants there is a trade-off 
for each item between the contribution value to the leader's objective 
and the chance of being included in the follower's solution set.

We analyze the leader's pricing problem for a natural greedy strategy of 
the follower and discuss the complexity of the corresponding 
problems. We show that setting the optimal weight values for \ll{} is, in 
general, \np-hard. It is even \np-hard to provide a solution within a 
constant factor of the best possible solution. Exact algorithms, based on 
dynamic programming and running in pseudopolynomial time, are provided. 
The additional cases, in which \ff{} faces a continuous (linear relaxation) 
version of the above problems, are shown to be straightforward to solve.
\end{abstract}

\section{Introduction}

A special class of decision problems in which two competing agents act sequentially 
are the so-called {\em Stackelberg  Games}, a name 
that has its origins in the early 1930s \cite{Stackelberg}. 
In particular one agent, the \textit{leader} \ll, needs to optimally choose the values 
of some variables that a second agent, the \textit{follower} \ff, uses as parameters 
of an optimization problem and the leader's payoff depends on the solution 
determined by the follower. These games can be formulated using bilevel 
programming  (BP).
In a bilevel optimization problem 
some of the constraints specify that a subset of variables must be an optimal 
solution to another optimization problem. This paradigm is particularly 
appropriate to model Stackelberg games.
From a computational point of view, bilevel programming is a hard task.
Jeroslow \cite{bib:j85} showed that even in presence of linear objective 
functions and constraints, BP is already \np-hard. The strong hardness of the same
problem was later proved by Hansen \textit{et al.} \cite{bib:hjs92}. Further
complexity results are presented in, e.g., \cite{bib:lv16} where the authors state
that most solution techniques for BP have been developed focusing on special cases 
in which convenient properties, such as linearity or convexity, can be exploited 
in order to develop efficient solution methods.

In the classical Subset Sum Problem (SSP), we are given a set of $n$ items 
$N=\{1,2, \dots ,n\}$, each having a positive integer weight $w_i$, $i=1, \dots , n$,  
and a knapsack of capacity $c$. The problem is to select a subset 
$A\subseteq N$ such that the corresponding total weight is closest to $c$ without 
exceeding $c$. 
SSP is a well known and widely studied problem and can be considered as a special 
case of the classical binary knapsack problem (KP) arising when the profit and 
the weight associated with each item are identical. Although SSP is a special case 
of KP, it is still {\np}-hard \cite{bib:gj79}, but can be solved to optimality in 
pseudopolynomial time and admits a FPTAS (see \cite{kpp04} for a review on solution 
algorithms for SSP). 
SSP has also been studied in the context of game theory  
\cite{bib:cms2015,bib:dnps14,bib:fj2013}, social choice, and multi-decider systems 
\cite{bib:npp2017,bib:sagt2015}.

%
%




In this paper, we consider a special Stackelberg game arising in a 
SSP setting: 
The leader \ll{} may alter 
the weights of some items (the leader's items $L\subseteq N$), while the 
follower \ff{} selects a solution set $A\subseteq N$ in order to 
utilize a bounded resource in the best possible way.  
In order to do so, \ff{} applies a simple, possibly suboptimal, strategy
which is known to \ll{}.
The leader receives a payoff only from its own items that are included in 
the solution, i.e., from items in $A\cap L$.
We address two variants of the problem: 
In a first version the leader  
may control, i.e., change, the weights of its items in the objective function 
(but not in the knapsack constraint) and wants to maximize the overall reduction
in the weights of its items included in the solution set as much as possible.
A second variant considers pricing items in the constraint only and 
not in the objective. 
In this case, the leader aims at maximizing the total weight increase
of its solution items.
In both variants there is a trade-off 
for each item between the contribution value to the leader's objective 
and the chances of being included in the follower's solution set.

Our study can be motivated by the application scenarios we sketch hereafter.
A financial trader \ll{} (e.g., a bank) offers a number of products (i.e., investments opportunities) to a client \ff{} who also has access to alternative forms of speculations but limited knowledge, bounded computational capacity, and restricted budget availability. 
All these products, either offered by the trader or not, are characterized 
by a given cost and ROI which are determined by the market and
the client selects a portfolio using simple strategies 
based on the efficiency of the financial products at hand. 
Due to client's limited information, the trader may offer its own products 
at a reduced ROI, taking the difference as a profit.
On the other hand, the trader has to make its investments attractive in order to 
increase chances that client purchases these instead of other products. 
In conclusion, the trader has to decide, for all its financial products, 
the (decreased) ROI offered to the client in order to maximize the overall profit 
which is in turn given by the overall discount on the actual ROI values associated 
to the products the trader has sold. 
In a slightly different setting, the trader tries to sell
its products at prices larger than the actual costs, while maintaining their real 
market profitability valid for the clients. In this case, the trader's 
profit is associated to the difference between the price set for those items included in the solution set and their genuine costs.
The model presented in this work represents the special case of a market 
scenario in which ROI is proportional to the cost of an investment.

A different example, motivating our model in a production planning context, considers 
an agent \ff{} that controls a processing resource with limited capacity 
(e.g., total available processing time) $c$. 
The agent performs its own jobs on the machine but also accepts external orders 
by a different agent \ll{} in order to 
improve the utilization of its resource. \ff{} wants to decide on a production 
plan satisfying the capacity constraint while maximizing
its utility which is 
directly proportional to the overall resource consumption.
In a first setting 
the leader wants to have some of its jobs processed by 
\ff{} but could also have the jobs processed at a different workshop, where it has 
to pay the same proportional price for each unit of processing time. \ll{} 
might gain something if it manages to have some of its jobs processed by \ff{} at a 
cheaper price. Even if less profitable, it may be beneficial for \ff{} to devote 
a fraction of---otherwise unused---resource capacity for processing the external 
jobs. In this case \ll{} wants to set prices of its jobs to maximize the total 
saving obtained by having its jobs processed by \ff{} at a cheaper cost than another
workshop.
An alternative situation 
arises if there is no possibility for \ll{} to negotiate prices, but \ff{} may be 
willing to let the leader use the machine for a bit longer than the original 
processing time of each job. Then the leader tries to extend this total extra time as much 
as possible (and use this time to do some additional work on these machines 
for free). 

\def\O{\mathcal{O}}

The paper is organized as follows. 
In the following Section~\ref{sec:problem} the addressed problem is rigorously defined.
In Section \ref{sec:objectivemodel} we consider the first variant of the problem 
where \ll{} is pricing items in the objective
and propose an algorithm that allows the leader to find the solution in pseudopolynomial time. 
We also show that, even though the problem admits a solution algorithm 
running in $\O(n c^2)$, it is not possible to obtain an  algorithm guaranteeing 
a constant as approximation ratio (unless $\mathcal{P}=\mathcal{NP}$).
Section \ref{sec:constraintmodel} addresses the second version of the problem 
where items are priced in the constraint. Even in this case  
constant-ratio approximation algorithms are not possible, however it is
possible to find an optimal pricing for \ll{} in a shorter time $\O(n^{\sfrac 3 2}c)$.
Oddly enough, a simple variant (quite similar at a first glance) of the latter problem
can be easily solved as discussed in Section~\ref{sec:variantconstraintmodel}.
Section \ref{sec:LPrelax} briefly sketches the straightforward solutions for the case in which the follower is allowed to include fractions of items in the solution set, i.e., \ff{} solves the LP relaxation for the above problems. 
Finally, in Section \ref{sec:concl}, some conclusions are drawn.

\subsection{Problem Definition}\label{sec:problem}

For ease of presentation, we partition the items into two classes: 
The set $L$ of items controlled by the leader \ll{} and the set $F = N\setminus L$ 
of the remaining items. For simplicity we also refer to $L$ and $F$ as the leader's 
items, or \ll-items, and follower's items or \ff-items.  
%
The original SSP posed to the agents can be written as an integer program: 
%
\begin{equation}\label{eq:originalPb}
    \max \left\{\sum_{j\in N}w_j y_j\ :\ \sum_{j\in N} w_j y_j \le c\,;\ y\in\{0,1\}^n \right\}
\end{equation}
As discussed above, the leader \ll{} is interested in filling the knapsack 
with its items (i.e., items $j\in L$) and, to this purpose, it may alter the 
parameters of Problem~\eqref{eq:originalPb} above in two different ways, 
by changing its item weights --as perceived by the follower-- either in the 
objective function or, alternatively, in the capacity constraint. Note that 
in this way we allow the item data to be separated and different in the 
objective and in the constraint.
Hereafter, we refer to the original and revised values of the leader's item $j$ 
by $w_j$ and $\tw_j$, respectively 
and limit our analysis to
nonnegative weights $\tw_j\geq 0.$

In this paper we consider the following two models:

\smallskip\noindent
{\em 1. Objective-control}: \ll{} may change the weights of its items {\em in the 
    objective} of Problem~\eqref{eq:originalPb}. Then \ff{} is confronted with the 
    following KP 
    \begin{align}\label{eq:RevisedObjPb}
    \begin{aligned}
    &    \max \left\{\sum_{j\in L} \tw_j y_j +\sum_{i\in F} w_i y_i\ :\ \sum_{j\in N} w_j y_j \le c\,;\ y\in\{0,1\}^n \right\}. &  
    \end{aligned}
    \end{align}
    Let $y^*$ be an optimal solution of~\eqref{eq:RevisedObjPb}. 
For each item $j\in L$, the leader \ll{} would like  to gain some of the value of $j$
by reducing the item profit (i.e., setting $\tw_j < w_j$) for the follower.
The resulting leader decision problem may be expressed via the bilevel program below: 
\begin{align}\label{eq:RevisedObjPbLeader}
\begin{aligned}
\max & \sum_{j\in L} (w_j - \tw_j) x_j  \\
\mbox{s.t. }  & x \in \arg\max_{{y \in \{0,1\}^{n}}} 
 \left \{ \sum_{j\in L} \tw_j y_j  + \sum_{i \in F} w_i y_{i} \ :  
 \sum_{j\in N} w_j y_j  \leq c \right\} \\
  & x \in \{0,1\}^{n};\ \tw\in\Re_+^{|L|}  
\end{aligned}
\end{align}
where $x$ and $\tw$ are the variables of the leader's 
problem: $x$ is restricted to be an optimal solution of  the follower's problem, 
i.e., $x$ components take the optimal values of the $y$ variables in the KP
faced by \ff. In particular, the binary variables $y_j$, in the follower's problem, 
indicate whether item $j$ is or is not included in the knapsack. Continuous variables 
$\tw_j$ indicate the chosen \ll-item weights. 

%
\smallskip\noindent 
{\em 2. Constraint-control}: \ll{} may change the weights of its items 
{\em in the constraint} of Problem~\eqref{eq:originalPb}. Then \ff{} is 
confronted with the following KP:
\begin{align}\label{eq:RevisedConstrPb}
\begin{aligned}
& \max \left\{\sum_{j\in N} w_j y_j\ :\ 
\sum_{j\in L} \tw_j y_j + \sum_{i\in F} w_i y_i \le c\,;\ y\in\{0,1\}^n \right\},&  
\end{aligned}
\end{align}
while \ll{} wants to maximize the weight increase of its items relative to their 
original values, in the optimal solution set of \ff{}. The corresponding
leader decision problem is therefore:
\begin{align}\label{eq:RevisedConstrPbLeader}
\begin{aligned}
\max & \sum_{j\in L} (\tw_j -w_j) x_j 
\\
\mbox{s.t. }  & x \in \arg\max_{{y \in \{0,1\}^{n}}} 
 \left \{\sum_{j\in N} w_j y_j\ :\ 
 \sum_{j\in L} \tw_j y_j  + \sum_{i \in F} w_i y_{i} \leq c \right\} \\
  & x \in \{0,1\}^{n};\ \tw\in\Re_+^{|L|}. &     
\end{aligned}
\end{align}
Here, increased $\tw_j$ values correspond to a higher utility 
for \ll{}. At the same time however, this lowers the chance that $j\in L$ is selected by 
the follower and thus becomes relevant for \ll. 

\smallskip
Note that, in both cases, the problem faced by the follower is a binary 
knapsack problem. KP is a (binary) \np-hard optimization problem, however 
we will assume that the follower has limited computational capability
as it was done in~\cite{bib:bhgv12}.
As a consequence, 
instead of finding the optimal solution, 
\ff{} determines a feasible solution set of the resulting knapsack problem 
by applying the very natural and intuitively appealing {\sc Greedy} 
strategy.
This holds for both versions of our problem, namely 
\eqref{eq:RevisedObjPbLeader} and \eqref{eq:RevisedConstrPbLeader}.
 
In KP the {\em efficiency} $e_j$, measured as the ratio between the 
profit and the weight of an item $j$, 
indicates how well $j$ utilizes the capacity in a solution set. 
{\sc Greedy} follows the natural idea of sorting the items in non-increasing order of 
efficiencies and then adding the items to the solution set in that order.
Whenever an item exceeds the current residual capacity it is discarded.
If an item can be added to the current solution it is inserted into the 
knapsack and never removed again. 
For the decision of \ff{} on a leader's item $j\in L$, in the first variant, 
i.e., Objective-control, we have 
$e_j = \sfrac {\tw_j} {w_j}$; while for the Constraint-control 
variant $e_j = \sfrac {w_j} {\tw_j}$. On the other hand, since all items $i \in F$ have the same efficiency $e_i = 1$,
it makes sense for a subset sum setting to apply a tie-breaking rule where 
the follower selects items of equal efficiency in decreasing order of weights.
With reference to the mechanism implemented by {\sc Greedy}, we 
write that an item $j$ is positioned {\em before} $i$ if $e_j > e_i$ (or,
{\em after} $i$ if $e_j < e_i$.)

%

In the next two sections, we consider the  integer problem and 
derive solution strategies for the leader.
In Section~\ref{sec:LPrelax}, 
we briefly consider a continuous model, in which \textit{fractions} of items are
allowed to be included  in the knapsack, and show that both versions of the 
problem are easy to solve in this case.

\section{Objective-control model}\label{sec:objectivemodel}


In general, \ll{} has the following options for each of its items:
$(ii)$ Either $j\in L$ is positioned before all {\ff}-items, i.e.\ with 
efficiency slightly larger than $1$ by setting $\tw_j := w_j + \eps$ and thus 
incurring a marginal loss for \ll, if the capacity permits it to be packed; or
$(ii)$ an item $j$ of $L$ is positioned after all {\ff}-items. In this case, 
the sorting by efficiencies serves only to implement a certain order of the 
remaining {\ll}-items offered to the {\sc Greedy} algorithm, which can also 
be reached by marginally small values of $\tw_j$. Thus, \ll\ gains (almost) 
$w_j$ in its objective for every packed item $j$. $(iii)$ A third option is to 
set $\tw_j := w_j$ thus inter-mixing leader and follower items. This choice 
cannot increase the range of possibilities for {\ll} or improve its objective either.
Hence, one may avoid considering strategies for leader that leave unaltered
the efficiencies of its items. As a consequence, an optimal 
solution can be partitioned into three (possibly empty) sets of items
selected by {\sc Greedy}: A first set of \ll-items, one block of
\ff-items, and a third set of \ll-items. 

Based on the above considerations, the optimal strategy for \ll\ 
is characterized by the total weight $W_1$ of all items placed before the {\ff}-items.
Recall that {\sc Greedy} selects the \ff-items---all with equal 
unit efficiency---in non-increasing order of weights and simply 
packs whatever items fit: let $F(W_1)$ be the total weight of 
these included items.  Now the residual capacity $\bar c = c-W_1-F(W_1)$ is 
available for packing with the remaining {\ll}-items.
The total weight of \ll-items packed in this step is denoted by $W_2 \leq \bar c$ 
and constitutes the objective value of \ll.
It is important to point out that, given the residual capacity $\bar c$, 
the leader may actually choose the set $S$ of its items that will be packed by
{\sc Greedy}: This can be done by slightly increasing the efficiency
of such items e.g., setting $\tw_j := \eps$ for each $j\in S$.  

Since the algorithm tries to pack each {\ff}-item, the residual capacity $\bar c$ 
remaining for \ll{} is upper bounded (unless the follower is left with no items 
at all) by the weight of the \emph{smallest unpacked} item of \ff{}.
In order to maximize the residual capacity for its own items, it would be natural to think that \ll{} would try to reach a capacity $\bar c$ strictly smaller than $w^F_{\min}$ for {\sc Greedy}.
However the following example
shows that, if \ff{} plays {\sc Greedy}, \ll{} may obtain better than that.
\begin{example}\label{ex:StructSol}
{\small
Consider an instance of our problem with $c=20$ and 8 items (sorted in non-increasing weight order) with  
$L = \{4, 5, 6, 8\}$ and $F = \{1, 2, 3, 7\}$. The item weights are 
$\{w_4 = 9, w_5 = 8, w_6 = 5, w_8 = 3\}$ for the \ll-items and  
$\{w_1 = 12, w_2 = 11, w_3 = 10, w_7 = 4\}$ for the items of {\ff}.

In Figure~\ref{fig:structsol}(a), a solution is depicted in which 
by sequencing $\langle w_4 = 9, w_5 = 8 \rangle$ we have 
$\bar c = 3$ which is the largest value smaller than $w^F_{\min} = 4$. \ll{} may obtain this 
solution by setting $\tw_4 = w_4 + \eps$ and $\tw_5 = w_5 + \eps$  (with cost $+2\eps$) 
so that $W_1 = 9+8 = 17$, $F(W_1) = 0$, and
$\bar c = 3$ that eventually \ff{} will pack with weight $w_8 = 3$ and \ll{} sets $\tw_8 = 0$. 
The leader objective value is $3-2\eps$. 
However, 
as illustrated in Figure~\ref{fig:structsol}(b), 
\ll{} may improve its objective by 
including items $w_5 = 8$ and $w_8 = 3$ with $W_1 = 11$, then \ff{} may only include $w_7 =4$ so, 
eventually, $\bar c = 5$ can be exploited by packing $w_6 = 5$ and the final value for the 
leader objective is $5-2\eps$.
\begin{figure}
    \centering
    \includegraphics[width=0.8\textwidth]{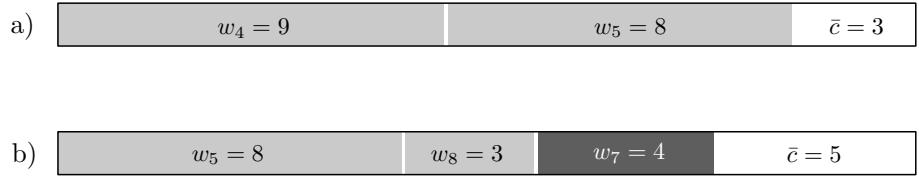}
    \caption{Structure of the solution sets with {\sc Greedy}}
    \label{fig:structsol}
\end{figure}
}\end{example}

From an algorithmic point of view\footnote{With no loss of generality, 
in the computations we neglect the $\eps$ values, needed  to ``guide'' {\sc Greedy} in 
the selection of the \ll-items.}, given $W_1$, the value $F(W_1)$ and thus 
$\bar c$ immediately follow and it remains for \ll{} to solve an instance 
of a standard SSP with capacity $\bar c$ and the item set without those 
items included in $W_1$. 
In order to do so, we have to consider all possible candidate values $W_1$ and
all weight values $W_2$ reachable after fixing $W_1$.
By  running a {\em dynamic programming by reaching} algorithm \cite{kpp04} and going 
through all $W_1$ values,the final optimal solution can be found easily 
in pseudopolynomial time.

\begin{theorem}\label{th:objectivepseudopol}
In the Objective-control Leader's problem \eqref{eq:RevisedObjPbLeader}, 
the optimal values $\tw_j$, $j\in L$ for \ll{} can be determined in time 
$O(n c^2)$. 
\end{theorem}
\begin{proof} In order to 
consider all possible candidate values $W_1$ and
we run a dynamic programming by reaching algorithm.

More formally, we define a two-dimensional array of reachable weight pairs
$r(W_1, W_2)$ for $W_1, W_2 \in \{0,\ldots, c\}$, where
$r(W_1, W_2)=1$ iff there exist two disjoint subsets 
$S_1, S_2 \subseteq L$ with $\sum_{i\in S_1} w_i =W_1$ and $\sum_{i\in S_2} w_i =W_2$,
and $r(W_1, W_2)=0$ otherwise.

The entries of this array can be found by running a {\em dynamic programming by reaching} algorithm as follows:
As an initialization we set all entries of $r$ to $0$ except $r(0,0)=1$.
Then we consider all items of \ll\ in turn (in arbitrary order).
For each item $j \in L$ we consider all pairs $(w', w'')$
with $w', w'' \in \{0,\ldots, c\}$:
If $r(w',w'')=1$, then we set $r(w'+w_j, w'')=1$ and $r(w', w''+w_j)=1$,
corresponding to the possibility of adding a weight $w_j$ before or after the {\ff}-items.

After $r$ is fully determined we go through all feasible choices of $W_1$,
i.e.\ all values $W_1$ where $\sum_{w''=0}^c r(W_1, w'') \geq 1$.
For each such candidate $W_1$ the value $F(W_1)$ is determined by executing {\sc Greedy} with item set $F$ and the value $\bar c =c-W_1-F(W_1)$ is computed.
Then, the best solution (given $W_1$) is determined as
$$W_2(W_1) :=\max \left\{ W_2 \mid W_2 \leq \bar c \mbox{ and } r(W_1,W_2)=1\right\}.$$
Finally, it remains to pick the optimal solution for \ll{} by taking the best choice for $W_1$ as
$$\max \left\{W_2(W_1) \mid \sum_{w''=0}^c r(W_1, w'') \geq 1\right\}.$$
The array $r$ has a size of $c^2$. Its computation requires each item in $L$ to be considered once for every array entry which gives a trivial pseudopolynomial running time bound of $O(n c^2)$. \end{proof}

\def\nnp{\mathcal{NP}}
Although an algorithm with pseudopolynomial running time exists, 
there is in fact no hope
to find an efficient algorithm that at least guarantees a constant 
approximation ratio, as shown by the following theorem.

\begin{theorem}\label{th:objectiveNoAPX} 
Unless $\mathcal{P}=\nnp$, the Objective-control Leader's 
Problem \eqref{eq:RevisedObjPbLeader} does not admit a constant 
approximation ratio.
\end{theorem}
\begin{proof}
Hereafter, we show that it is \np-complete to determine  weights $\tw_j$ 
in order to  obtain  a solution which is within a constant factor of the best possible solution for \ll .
Consider an instance $I$ of the decision problem {\sc Partition} \cite{bib:gj79}, 
where $m$ integer numbers $a_1, \ldots, a_m$ are given
and the question is whether there exists a subset of these 
numbers with total sum equal to 
$b = \frac 1 2 \sum_{i=1}^m a_i$.

Starting from $I$, we define the following instance $I'$ of the Objective-pricing Leader's Problem, 
where $M$ is a large enough number: 
The leader \ll{} has $m+1$ items with profits corresponding to the values of instance $I$, 
i.e.\ $w_j=a_j$, and one additional item with weight $w_{m+1}=M$.
The follower \ff{} has only one item of weight $w_{m+2}=M+1$.
The knapsack capacity is set to $c=M+b$.

As observed above, \ll{} can only place its items before the \ff-item, 
with a marginal positive cost in the objective, or 
after the \ff-item, with a utility equal to their full weight. 

If $I$ is a YES-instance of {\sc Partition}, then there exists a subset $S$ with $\sum_{i \in S} a_i= b$.
The optimal strategy of \ll{} will submit this set $S$ first, before the {\ff}-item of weight $M+1$.
Now the {\ff}-item $m+2$ is blocked since $b+M+1 > c$.
However, the \ll-item of weight $w_{m+1}$ can still be packed 
and, by setting $\tw_{m+1} := 0$, \ll{} gets a ``gain'' equal to $M$.

Otherwise, if $I$ is a NO-instance, there are two cases to consider:
\begin{enumerate}
\item \ll{} lets \ff{} packs its item:
Hence, the optimal strategy of \ll{} will not submit any items before $w_{m+2}$ thus setting $W=0$ and gaining at most $b-1$.

\item \ll{} blocks the \ff-item by submitting, before $m+2$,  a subset of items  with total weight $W\geq b+1$.
Then the follower item ${m+2}$ is blocked but also ${m+1}$ of \ll{} with weight $w_{m+1} = M$ does not fit anymore. 
However, $L$ can submit all its remaining items and gets them packed by {\sc Greedy}. 
This results in a gain of at most $2b-W\leq b-1$.
\end{enumerate}
Summarizing, the instance $I$ of {\sc Partition} has a YES answer if and only if
the optimal weight obtained by the leader in the Stackelberg subset sum problem is 
equal to $M$ while it is less than $b$ for any NO answer.
This rules out a polynomial time approximation algorithm with any constant 
approximation ratio $\rho$, since we can choose, e.g., $M >  \rho\, b$. \end{proof}

\section{Constraint-control model}\label{sec:constraintmodel}

We now turn to the Constraint-control model  illustrated by
Problem~\eqref{eq:RevisedConstrPbLeader}.
 %
In this case the coefficients in the objective function for 
\ff{} are fixed and equal to the original given weights $w_j$. 
Instead, the leader can modify the weights of the {\ll}-items in 
the constraint. 
%

It is not hard to see that, for this version of the problem, the solution 
structure is somehow similar to that of Program \eqref{eq:RevisedObjPbLeader}, 
however this variant is slightly less complex to solve.  
As before, recalling that here the efficiency of an item $j$ in the 
follower's problem is $e_j = \sfrac{w_j}{\tw_j}$, the leader can 
position any item $j \in L$ before the {\ff}-items 
by accepting a marginal loss, i.e., setting $\tw_j := w_j - \eps <0 $.
Denote the total weight of these items again by $W_1\leq c$.
Then, \ff{} will include some of its items following {\sc Greedy}
with total weight $F(W_1)$.
Finally, when \ff{} has no more items to include, the leader may set 
its objective value using the residual capacity $\bar c = c-W_1-F(W_1)$.
It suffices to pick the smallest remaining {\ll}-item of weight $w'$ 
(if it fits) and increase its weight up to $c-W_1-F(W_1)$ 
so that the capacity is fully consumed. All the remaining \ll-item weights should  
be increased to a large enough value (e.g., $\bar c + 1$) in order to guarantee 
that {\sc Greedy} actually picks the desired item $w'$.

A straightforward algorithm for determining the optimal solution for \ll{} 
would consider each item $w' \in L$ as a candidate for the smallest remaining item.
For each of these $|L|$ choices, one can run a simple one-dimensional
{\em dynamic programming by reaching} 
to determine all weight values $W_1$ which can be reached by a subset of $L \setminus \{w'\}$.
By a suitable implementation of the dynamic programming algorithm, we may state:

\begin{theorem}\label{th:constraintpseudopol}
In the Constraint-control Leader's problem \eqref{eq:RevisedConstrPbLeader}, 
the optimal weights $\tw_j$ for $j\in L$ can be determined in pseudopolynomial 
time of $O(n^{\sfrac 3 2} c)$.
\end{theorem}
\begin{proof}
For each possible candidate as smallest remaining item $w'\in L$,
we may run a dynamic programming by reaching algorithm in order to
determine all weight values $W_1$ that can be achieved by a subset
of $L\setminus\{w'\}$.
This can be done in $O(n c)$ time e.g.\ using a one-dimensional reduction 
similar to the algorithm devised for the objective-control model.
For every reachable value $W_1$, {\sc Greedy} can be applied for the 
\ff-items yielding $F(W_1)$ and requiring at most $O(n c)$ time (after 
sorting the \ff-items once). If $\bar c = c-W_1-F(W_1) > w'$, we keep
$\bar c - w'$ as a candidate for the objective function value.

In this way we consider every item in $L$ as a potential candidate for 
$w'$ and finally keep the best objective function value for the optimal 
solution. All together, this yields an $O(n^2 c)$ pseudopolynomial 
running time.

\smallskip
It is possible to improve the running time of the algorithm by considering 
a slightly more involved implementation of the dynamic programming algorithm 
which works as follows (we omit the details on the rounding procedure to obtain 
integer valued quantities when required): Given a parameter $k$ (to be defined 
later) the set $L$ is partitioned into $k$ subsets $L_i$, $i =1,\ldots, k$, of 
(roughly) equal size, i.e.\ $|L_i| \approx |L|/k$. Instead of running the dynamic 
programming iteration separately for each item $w' \in L$, we perform the computation 
in $k$ phases, considering all items $w' \in L_i$ in one phase.

In each such phase we first compute an initialization of the dynamic programming 
array taking all items in $L \setminus L_i$ into account which takes $O(n c)$ 
time. This array remains unchanged and is stored until the completion of the 
current phase. Then we pick one candidate $w' \in L_i$ and continue the computation 
from the above initialization by adding all items in $L_i \setminus \{w'\}$.
This takes $O(n/k \cdot c)$ time and gives all values $W_1$ for the chosen $w'$.
Iterating this process for all candidates in $L_i$ (going back to the precomputed 
initialization of the array for each candidate) adds a total running time of 
$O((n/k)(n/k) \cdot c)$ for this phase.

Iterating this approach over all sets $L_i$ produces the dynamic programming array 
for all candidates $w' \in L$ and takes $O(k (n\cdot c + n^2 / k^2 \cdot c))$ time. 
Choosing the parameter $k \approx \sqrt{|L|}$ we get an overall running time
of $O(n^{\sfrac 3 2} c)$.

It remains to consider the outcome of {\sc Greedy} and the resulting 
residual capacity $\bar c$, resp.\ the difference $\bar c -w'$, for 
each candidate $w'$. To do so, we can precompute the result of {\sc Greedy} 
for every capacity value $W_1=0,1,\ldots, c$ in $O(n\cdot c)$ time 
(after sorting the items of $F$) and store the outcome in an array $F(W_1)$.
Now it is easy to determine the value $\bar c-w'$ in constant time 
whenever we detect a new reachable value $W_1$ for some candidate $w'$ and 
check for a new objective function value. \end{proof}

The existence of a pseudopolynomial algorithm which is just linear in the 
capacity $c$ would justify an even stronger hope for positive approximation 
results than for Objective-pricing. In fact, many discrete optimization 
problems with a linear pseudopolynomial algorithm permit even an FPTAS by 
rounding the objective value space. 
However, the following theorem shows that those expectations cannot be fulfilled.

\begin{theorem}\label{th:constraintNoAPX} 
Unless $\mathcal{P}=\nnp$, the Constraint-control Leader's 
Problem \eqref{eq:RevisedConstrPbLeader} does not admit a constant 
approximation ratio.
\end{theorem}
\begin{proof}
Consider again an instance $I$ of {\sc Partition},
where $m$ integer numbers $a_1,  \ldots, a_m$ are given
and the question is whether there exists a subset of these 
numbers with total sum equal to $b = \frac 1 2 \sum_{i=1}^m a_i$.

In the following we will construct an instance of the Constraint-pricing model such that \ll{} receives an objective function value of (almost) $1$ if $I$ is a YES-instance and $0$ otherwise.
Thus, any polynomial algorithm guaranteeing a constant approximation ratio for the Constraint-pricing Leader's problem would also answer the decision problem for $I$ in polynomial time.

Item set $L$ consists of $m$ items with $w_j = 2 a_j$ and another item $w_{m+1}=\eps$.
Let $k:= \min\{j : 2^j > \sum_{i=1}^{m+1} w_i \}$ and add
one additional item $w_{m+2} = 2^k - 2b$. 
Thus, if $I$ is a YES-instance then there is also a subset of $L$
with total weight $2^k$. 
On the other hand, if $I$ is a NO-instance then no subset of $L$ can have a weight of $2^k$ 
(consider that $w_{m+2}$ must be necessarily included in such a subset since $\sum_{i=1}^{m+1} w_i < 2^k$).

Item weights in $F$ will be denoted for convenience by $v_i$
and $F$ consists of $k-1$ items with weights 
$v_i=2^i$ for $i=1,\ldots, k-1$.
Clearly, the subsets of these items can reach every even weight value between $2$ and $2^k-2$.
Then we add an item $v_k=3$ which means that now every weight value (even and odd) between $2$ and $2^k+1$ can be reached.
Finally, we add a large item $v_{k+1}=2^k+2$.
This means that there are subsets of $F$ with total weight for every value from $2$ to $2^{k+1}+3$ except for the missing value $2^k+3$
which is larger than $\sum_{i=1}^k v_i$ and larger than $v_{k+1}$ but smaller than $v_{k+1}+v_1$.

Now we choose $c=2^{k+1}+3$.
If $I$ is a YES-instance, we argued above that \ll{} can submit items before the {\ff}-items with total weight $2^k$ which leaves a residual capacity of $2^k+3$ for \ff. 
As shown above, there is no way for \ff{} to fill this capacity completely but a residual capacity of $\bar c=1$ will remain which \ll{} can fill with item $w_{m+1}$ thus gaining an objective function value of $1-\eps$.

If $I$ is a NO-instance, then \ll{} can place some subset of items before the {\ff}-items with total weight different from $2^k$ leaving a gap different from $2^k+3$. 
\ff{} can fill any such gap completely with its items leaving an objective value of $0$ for \ll.
(In fact \ff{} could not fill a gap of only $1$,
but \ll{} can not enforce a gap of $1$ since its sum of weights is $< 2^{k+1}$.)

Observe that the above reduction is polynomial in the length of the encoded input since $\sum_{i=1}^{m+1} w_i \leq 2m \cdot a_{\max}$ and thus $k$ can be expressed as a logarithm of $m \cdot a_{\max}$. 


\end{proof}

\subsection{A simple variant}\label{sec:variantconstraintmodel}

A variant of the Constraint-pricing model in which the leader's objective is   
given by the actual weight values of its items can also be addressed. 
As we see below, this problem turns out to be easy to solve.
The bilevel program associated with the leader problem is the following, 
with the usual variables:

\begin{align}
\begin{aligned}
\max & \sum_{j\in L} \tw_j  x_j \label{eq:bilevel_1_constraint_1} & \\
\mbox{s.t. } & x \in \arg\max_{{y \in \{0,1\}^{n}}} 
\left \{ \sum_{j\in N} w_j y_{j} \ : \ 
\sum_{j\in L} \tw_j y_j  + \sum_{i \in F} w_i y_{i} \leq c \right\} \\
& x \in \{0,1\}^{n};\ w\in\Re^{|L|}
\end{aligned}
\end{align}

As we observed above, if \ll{}  decreases an item weight, i.e.\ $\tw_j < w_j$, 
then $j$ is positioned before all {\ff}-items since it has an efficiency  larger than $1$.
In this case, reducing the weights by a marginal quantity gives \ll{} control 
over the items considered by \ff{} but with a minimal loss in the leader objective.
So, by setting $\tw_j := w_j - \eps$, and thus incurring a marginal loss for \ff, 
item $j$---if the capacity allows it---will be included in the solution.
Clearly, a \ll-item $j$ is positioned after all {\ff}-items if $\tw_j > w_j$.

This allows a very simple solution strategy: {\ll} positions all its items before 
{\ff} in arbitrary order. As soon as an item can not be packed, {\ll} reduces its 
weight to match exactly the remaining capacity. Thus, neglecting the marginal loss 
due to the $\eps$, {\ll} is guaranteed a best possible objective function value 
$\min\{c, \sum_{j\in L} w_j\}$. 

\section{LP relaxation}\label{sec:LPrelax}

We now consider the special case in which \ff{}  faces a continuous (linear 
relaxation) version of the above problems \eqref{eq:RevisedObjPbLeader} and
\eqref{eq:RevisedConstrPbLeader} so that variables $x$ and $y$ are continuous
and contained in the interval $[0,1]$. \ff{} is therefore able to optimally 
solve the follower problem in polynomial
(in fact linear) time.
We now show that in this case both problems are trivial. 

Let us first consider the {Objective-control} model \eqref{eq:RevisedObjPbLeader}. 
Since \ff{} can split items to be included in the solution set, 
if $\sum_{j \in F} w_j \geq c$ then the follower can fill the entire knapsack capacity 
with its own items (or some of \ll{} with $\tw_j = w_j$)
and \ll{} cannot gain anything. 
%
However, if $w(F) < c$, 
i.e., the capacity exceeds the total weight of the {\ff}-items,
then \ll{} can set $\tw_j := 0$ for all $j\in L$, and pack 
as much as possible of its items in the residual capacity.
In conclusion, the optimal solution value for \ll{} is always given by 
$\min\left\{0, c-w(F)\right\}$.

Now let us consider the Constraint-control model \eqref{eq:RevisedConstrPbLeader}.
If $\sum_{j \in F} w_j \geq c$ then the follower will again fill the entire knapsack capacity with its own items.
\ll{} could only set $\tw_j := w_j - \eps$ to have an item included in the knapsack before the \ff{}-items, which would yield a negative contribution.
Thus, \ll{} remains at value $0$.
If $\sum_{j \in F} w_j < c$ then 
\ll{} can pick an arbitrary item $j' \in L$,
preferably the item with minimal weight, 
and set $\tw_{j'}:=M$.
All other items $j\in L$ receive values $\tw_j$ yielding an even lower efficiency.
Now after packing all items in $F$, \ff{} will choose a fractional part of $j'$ by setting $y_{j'}=\frac{c-w(F)}{M}$ to fill the knapsack completely.
The gain for \ll{} is $(M-w_{j'})\frac{c-w(F)}{M}$ which tends to
$c-w(F)$ for $M\to\infty$.
Hence, the optimal solution value for \ll{} is again 
$\min\left\{0, c-w(F)\right\}$.

\smallskip
For variant  \eqref{eq:bilevel_1_constraint_1} of the Constraint-pricing model, 
in which the leader is interested in maximizing $\sum_{j\in L} \tw_j$, 
the continuous model works exactly in the same way as the discrete case treated in Section~\ref{sec:variantconstraintmodel}.
Only for the last item of \ll{} which cannot be packed completely,
the continuous model does not require any special weight selection
but guarantees the objective function value $\min\{c, \sum_{j\in L} w_j\}$
by default (disregarding $\eps$).

\section{Conclusions}\label{sec:concl}

In this paper we analyzed the complexity of a Stackelberg  game for a 
Subset Sum pricing problem. 
We considered two variants of the problem 
in which the leader may revise the items weight in the follower objective or in the knapsack
(capacity) constraint and the follower selects the solution items set using the natural {\sc Greedy} algorithm heuristic. 
The objective function of the leader relates to the overall variation operated on the items weight. 
We showed that both versions of the problem are binary \np-hard but can be solved by dynamic programming in pseudopolynomial time.
Even though, both versions turn out to be non approximable within a constant factor.
We also characterize some easy cases and show that the continuous 
relaxation versions of the problems permit straightforward solution procedures.


There are a number of  open questions directly addressing the results presented here. 
In particular, a natural generalization of the addressed problem involves
the binary knapsack problem, so that items are characterized by their 
weight and profit parameters.
We also assumed that \ff{} adopts a computationally-bounded strategy.
Removing this 
constraint and considering a follower able to optimally solve the resulting
optimization problem makes the leader strategy much harder to be determined.

\subsection*{Acknowledgements}

 %

  Ulrich Pferschy was supported by the project ``Choice-Selection-Decision" and by the COLIBRI Initiative of the University of Graz.

\end{document}